\documentclass[11pt,english]{article}

\usepackage{times}
\usepackage{amsfonts}
\usepackage{amsmath,amssymb}
\usepackage{amsthm}
\usepackage{natbib}
\usepackage{bm}
\usepackage{graphicx}
\usepackage{subcaption}
\usepackage{appendix}
\usepackage{rotating}
\usepackage{float}
\usepackage{authblk}
\usepackage{color}
\usepackage[colorlinks,linkcolor=blue,citecolor=blue,bookmarks=false,pagebackref]{hyperref}
\usepackage[top=1.0in, bottom=1.3in, left=1.15in, right=1.15in]{geometry}
\usepackage{setspace}
\usepackage{titling}
\allowdisplaybreaks

\theoremstyle{plain}
\newtheorem{theorem}{Theorem}[section]
\theoremstyle{plain}
\newtheorem{lemma}{Lemma}[section]
\theoremstyle{plain}

\theoremstyle{plain}

\theoremstyle{plain}
\newtheorem{corollary}{Corollary}[section]
\theoremstyle{remark}

\theoremstyle{definition}

\newcommand{\argmin}{\operatornamewithlimits{argmin\,}}

\newcommand{\was}{\mathcal{W}}
\newcommand{\prox}{\text{prox}}

\mathchardef\mhyphen="2D

\newtheorem{innercustomthm}{Theorem}
\newenvironment{customthm}[1]
  {\renewcommand\theinnercustomthm{#1}\innercustomthm}
  {\endinnercustomthm}

\floatstyle{ruled}
\newfloat{algorithm}{tbp}{loa}
\providecommand{\algorithmname}{Algorithm}
\floatname{algorithm}{\protect\algorithmname}

\title{Langevin Monte Carlo and JKO splitting}
\author{Espen Bernton\thanks{Department of Statistics, Harvard University, USA.
Email: ebernton@g.harvard.edu.}}
\date{}

\begin{document}

\maketitle

\begin{abstract}
\noindent Algorithms based on discretizing Langevin diffusion are popular tools for sampling from high-dimensional distributions. We develop novel connections between such Monte Carlo algorithms, the theory of Wasserstein gradient flow, and the operator splitting approach to solving PDEs. In particular, we show that a proximal version of the Unadjusted Langevin Algorithm corresponds to a scheme that alternates between solving the gradient flows of two specific functionals on the space of probability measures. Using this perspective, we derive some new non-asymptotic results on the convergence properties of this algorithm.
\end{abstract}

\section{Introduction}

In this paper, we shed new light on Langevin-based Monte Carlo algorithms by drawing connections to the Wasserstein gradient flow literature and the operator splitting approach to solving PDEs. In a seminal paper, \citet{jordan1998variational} expressed the solution of the Fokker--Planck equation as the gradient flow of the relative entropy functional (otherwise known as the KL-divergence) with respect to the $2$-Wasserstein distance. Their constructive proof used a time discretization approach that has since become known as the JKO scheme. We show that applying the JKO scheme in conjunction with a splitting approach to solving the Fokker--Planck equation reduces to a proximal version of the Unadjusted Langevin Algorithm. Our proofs rely heavily on the theory developed by \citet{ambrosio2005}, and have the benefit of holding for potentials that are not necessarily differentiable. In turn, this allows us to provide some new results regarding the convergence of the algorithm. Our work is related to \citet{durmus2016efficient}, and we will make comparisons to their theoretical results.

To motivate the use of Langevin-based Monte Carlo algorithms, consider a log-concave target distribution $\pi$, given in terms of the Lebesgue density $\pi(x) = Z^{-1}e^{-V(x)}$, where $V:\mathbb{R}^d \to \mathbb{R}$ is a convex function, $d\in \mathbb{N}$ is an integer, and $Z$ is the normalizing constant. In the case where $V$ is differentiable, we can associate with it the Langevin diffusion, given in terms of the It\^o stochastic differential equation
\begin{equation}\label{eq:langevin}
dX(t) = -\nabla V(X(t))dt +  \sqrt{2}dW(t), \quad X(0) = X_0 \sim \rho_0.
\end{equation}
It represents the position $X(t) \in \mathbb{R}^d$ of a particle at time $t >0$, initialized at the random location $X_0 \sim \rho_0$, with drift according to the gradient of the potential $V$ and subject to random perturbations $dW(t)$. The process $W(t)$ is the standard Wiener process. The density of $X(t)$ at time $t$, written $\rho(t)$, satisfies the linear Fokker--Planck equation:
\begin{equation}\label{eq:fp}
\frac{d\rho}{dt} = \text{div}(\rho \nabla V) + \Delta \rho, \quad \rho(0) = \rho_0.
\end{equation}

A classical result says that under quite weak convexity and smoothness conditions on $V$, the unique stationary solution of \eqref{eq:fp} is equal to $\pi$, and that convergence to $\pi$ is exponentially fast \citep[see for example][Chapter 4]{pavliotis2014stochastic}. These attractive properties have spawned a range of sampling algorithms targeting $\pi$ based on time discretizations of the process in \eqref{eq:langevin}. Notably, the Unadjusted Langevin Algorithm (ULA) and its Metropolis adjusted counterpart MALA have received much attention.

The Unadjusted Langevin Algorithm is simply an explicit Euler discretization of \eqref{eq:langevin}: for a time-step $h> 0$ and for $k \geq 0$,
\begin{equation}
X_h^{k+1} = X_h^{k} - h\nabla V(X_h^{k}) + \sqrt{2h}\eta^{k+1}, \quad X_h^{0} = X_0,
\end{equation} 
where $(\eta^{k})_{k\geq 1}$ is a sequence of independent $\mathcal{N}(0,\mathcal{I}_d)$ random variables and $\mathcal{I}_d$ is the $d$-dimensional identity matrix. In MALA, $X_h^{k+1}$ is either accepted or rejected in a Metropolis step with the purpose of removing the asymptotic bias of ULA stemming from discretization error.

Originating with \citet{roberts1996exponential}, there has been a lot of interest in quantifying the performance of these algorithms, with early work primarily focusing on MALA \citep[see e.g.][]{jarner2000geometric,roberts2002langevin,pillai2012optimal,xifara2014langevin}. It was not until \citet{dalalyan2014theoretical}, who gave precise bounds for the total variation distance between the law of $X_h^{k}$ and $\pi$ in terms of $d, k$, and $h$, that ULA garnered similar attention. His results were further improved and extended to other metrics and discrepancies by \citet{durmus2016sampling,durmus2017nonasymptotic, cheng2017convergence, dalalyan2017further}. For instance, \citet{dalalyan2017user} show that if $V$ is strongly convex and has Lipschitz continuous gradient, then  $\Omega(d/\varepsilon^2)$ iterations are sufficient for ULA to achieve an error of $\varepsilon$ in the $2$-Wasserstein distance. Similar results also hold in situations where only a (sufficiently regular) approximation of the gradient is available.

In what follows, we will view Langevin-based Monte Carlo through the lens of Wasserstein gradient flow, and show that this perspective can lead to interesting results on the computational complexity of such algorithms. Wasserstein gradient flow was also used by \citet{cheng2017convergence} as a theoretical tool to study ULA, but our approach makes closer connections to the operator splitting literature, and as such leads to different results. We hope that further connections can have methodological implications in these fields, by considering the wide variety of JKO schemes, splitting schemes, and Langevin Monte Carlo algorithms that exist. 

The rest of this paper is structured as follows. Section \ref{sec:notation} defines the notation and states some important definitions, Section \ref{sec:wasserstein_gradient_flow} reviews some concepts from the Wasserstein gradient flow literature, Section \ref{sec:operator_splitting} briefly discusses the operator splitting approach to solving PDEs, Section 4 establishes connections between Wasserstein gradient flow, operator splitting and Langevin Monte Carlo and includes some convergence results on the proximal version of the ULA algorithm, and Section 5 concludes. Proofs are given in the Appendix.

\subsection{Notation and definitions}\label{sec:notation}
Let $\|\cdot\|_p$ be the $\ell_p$-norm on $\mathbb{R}^d$, unless $p = 2$, in which case it reduces to the Euclidean distance and is denoted by $\|\cdot\|$. Define $\mathcal{P}_2(\mathbb{R}^d)$ to be the set of probability measures on $\mathbb{R}^d$ with finite second moments with respect to the Euclidean distance. The $2$-Wasserstein distance is a metric on $\mathcal{P}_2(\mathbb{R}^d)$, and is for any $\mu,\nu \in \mathcal{P}_2(\mathbb{R}^d)$ defined by
\begin{equation}
\was_2(\mu,\nu) = \left(\inf_{\gamma \in \Gamma(\mu,\nu)} \int_{\mathbb{R}^d\times \mathbb{R}^d} \|x-y\|^2 d\gamma(x,y)\right)^{\frac{1}{2}},
\end{equation}
where $\Gamma(\mu,\nu)$ is the set of all joint distributions with marginals $\mu$ and $\nu$. A desirable feature of the $2$-Wasserstein distance is that 
$\was_2(\mu_n, \mu) \to 0$ as $n \to \infty$ if and only if $\mu_n$ converges weakly to $\mu$ and the corresponding sequence of second moments also converges \citep[Theorem 6.9]{villani2008}.

The entropy and potential energy functionals, $\rho \mapsto \mathcal{H}(\rho)$ and $\rho \mapsto \mathcal{V}(\rho)$  respectively, are given by 
\begin{equation}\label{eq:entropy}
\mathcal{H}(\rho) =
\begin{cases}
\int \log \rho d\rho & \text{for $\rho \ll \mu_{\text{Leb}}$},  \\
+\infty & \text{otherwise},
\end{cases}
\end{equation}
where $\mu_{\text{Leb}}$ denotes the Lebesgue measure on $\mathbb{R}^d$, and
\begin{equation}\label{eq:potential_energy}
\mathcal{V}(\rho) = \int V d\rho.
\end{equation}
The relative energy functional $\rho \mapsto \mathcal{H}(\rho | \pi)$, also called the KL-divergence, is given by
\begin{equation}\label{eq:relative_entropy}
\mathcal{H}(\rho | \pi)  =  \mathcal{H}(\rho) + \mathcal{V}(\rho) + \log Z.
\end{equation}

An important concept in optimal transport, which will play a significant role later, is the notion of displacement convexity. A functional $\rho \mapsto \mathcal{F}(\rho)$ is
said to be $\lambda$-displacement convex for some $\lambda \in\mathbb{R}$ if, for all $t\in [0,1]$,
\begin{equation}
\mathcal{F}(\mu_t) \leq (1-t)\mathcal{F}(\mu_0) + t\mathcal{F}(\mu_1) - \frac{\lambda}{2}t(1-t)\was_2^2(\mu_0,\mu_1)
\end{equation}
for any constant speed geodesic $\mu:[0,1]\to \mathcal{P}_2(\mathbb{R}^d)$. A curve $\mu:[0,1]\to \mathcal{P}_2(\mathbb{R}^d)$ is a  constant speed geodesic if, for any $0\leq s\leq t\leq 1$, we have that $\was_2(\mu_s,\mu_t) = (t-s)\was_2(\mu_0,\mu_1).$

We use the following notation for the density of a Gaussian distribution with zero mean and covariance matrix $2t\mathcal{I}_d$:
\begin{equation}\label{eq:gaussian_kernel}
\phi_t(x) = \frac{1}{(4\pi t)^{d/2}}\exp\left(- \frac{\|x\|^2}{4t}\right).
\end{equation}

By a Markov operator, we mean a linear functional $R$ that maps the set of non-negative Lebesgue integrable functions into itself. A family of Markov operators $(R_t)_{t\geq0}$ is called a Markov semigroup if $R_0$ is the identity map, $R_{t+s} = R_tR_s$ for any $s,t\geq 0$, and the map $t\mapsto R_t f$ is continuous for any non-negative and Lebesgue integrable $f$.

\section{Wasserstein gradient flow} \label{sec:wasserstein_gradient_flow}
The theory of gradient flows in the space of probability measures was pioneered by Ambrosio, Gigli and Savar\'e in their book \citet{ambrosio2005},  generalizing the variational structure \citet{jordan1998variational} had used to describe the diffusion and Fokker--Planck equations. With Langevin Monte Carlo in mind, we provide only a brief introduction to this theory, and refer to the aforementioned references and the accessible review of \citet{santambrogio2016euclidean} for further details.

We first consider continuous time flows, which will lead to a useful perspective on generalizations of the continuous time processes in \eqref{eq:langevin} and \eqref{eq:fp}. Secondly, we consider the time discretizations through which the existence and uniqueness of gradient flows are typically established. Although they were originally introduced as theoretical tools in the literature, it will later become clear that Langevin Monte Carlo in fact numerically approximates such a time discretization.

\subsection{Continuous time flows}
In Euclidean space, a curve $x:[0,\infty) \to \mathbb{R}^d$ is the gradient flow, or steepest descent, of a differentiable function $f:\mathbb{R}^d \to \mathbb{R}$ if
\begin{equation}\label{eq:grad_flow_euclidean}
\frac{dx}{dt} = -\nabla f(x), \quad {x(0) = x_0}.
\end{equation}
By analogy, one can interpret the gradient flow of a functional $\mathcal{F} :\mathcal{P}_2(\mathbb{R}^d) \to \mathbb{R}$ to be a curve $\rho:[0,\infty) \to \mathcal{P}_2(\mathbb{R}^d)$ that satisfies
\begin{equation}\label{eq:grad_flow_wasserstein}
\frac{d\rho}{dt} = -\nabla_{\was_2} \mathcal{F}(\rho), \quad {\rho(0) = \rho_0},
\end{equation}
for some generalized notion of gradient $\nabla_{\was_2}$, in terms of the $\was_2$ metric. For sufficiently regular $\rho$ and $\mathcal{F}$, $\nabla_{\was_2} \mathcal{F}(\rho)$ corresponds to  $-\text{div}(\rho \nabla\frac{ \delta\mathcal{F}}{\delta \rho})$, where $\delta\mathcal{F} /\delta \rho$ is the first variation of $\mathcal{F}$. Applied to the functional of interest, namely $\mathcal{F}(\rho) = \mathcal{H}(\rho | \pi)$, one has that $\delta\mathcal{F} /\delta \rho = V + \log \rho + 1$. Thus, if $V$ is differentiable one recovers \eqref{eq:fp} \citep[see e.g.][Lemma 10.4.1]{ambrosio2005}.

Due to the technically challenging nature of defining Wasserstein gradients this way when $V$ is not differentiable, we instead adopt the definition given in \citet{ambrosio2009existence}, inspired by the characterization of gradient flows in terms of evolution variational inequalities (EVIs) shown in \citet[][Theorem 11.1.4]{ambrosio2005}. In particular, we say that a continuous curve $\rho : (0,+\infty) \to \mathcal{P}_2(\mathbb{R}^d)$ is a gradient flow of a $\lambda$-displacement convex functional $\mathcal{F}$ if
\begin{equation}\label{eq:grad_flow}
\frac{d}{dt} \frac{1}{2} \was_2^2(\rho(t),\nu) + \frac{\lambda}{2}\was_2^2(\rho(t),\nu) + \mathcal{F}(\rho(t)) \leq \mathcal{F}(\nu),
\end{equation}
holds in the sense of distributions, for all $\nu\in \mathcal{D}(\mathcal{F}) = \{\mu \in \mathcal{P}_2(\mathbb{R}^d) : \mathcal{F}(\mu) < +\infty\}$. The flow is said to start from $\rho_0$ if $\was_2(\rho(t),\rho_0) \to 0$ as $t \to 0$. Here, ``in the sense of distributions'' means that for all infinitely differentiable and compactly supported test functions, denoted $f \in C_c^\infty((0,\infty); \mathbb{R})$, such that $f\geq 0$, we have
\begin{equation} \label{eq:sense_of_distributions}
-\frac{1}{2}\int_0^\infty\was_2^2(\rho(t),\nu)f^\prime(t) dt \leq \int_0^\infty \left[\mathcal{F}(\nu)- \mathcal{F}(\rho(t)) - \frac{\lambda}{2}\was_2^2(\rho(t),\nu) \right]f(t) dt.
\end{equation}
The connection between \eqref{eq:grad_flow} and \eqref{eq:sense_of_distributions} can be seen by imagining the left hand side of \eqref{eq:sense_of_distributions} being integrated by parts. 

One of the most attractive features of gradient flows are their convergence properties. For any $\lambda$-displacement convex functional $\mathcal{F}$ with $\lambda>0$, the map $\rho \mapsto \mathcal{F}(\rho)$ has a unique minimum $\bar{\rho}$, and Theorem 11.2.1 of \citet{ambrosio2005} states that there exists a unique gradient flow $t\mapsto \rho(t)$, which satisfies
\begin{equation}\label{eq:contraction_F}
\was_2(\rho(t),\bar{\rho}) \leq \was_2(\rho_0,\bar{\rho})e^{-\lambda t} \quad \text{and}\quad \mathcal{F}(\rho(t)) - \mathcal{F}(\bar{\rho})  \leq \left[\mathcal{F}(\rho_0) - \mathcal{F}(\bar{\rho})\right]e^{-2\lambda t},
\end{equation}
or any $t\geq 0$. Convergence results also exist in the case where $\lambda = 0$, but do not yield the exponential convergence observed above. 

This result can be applied to the relative entropy by making the following observations: when $V$ is $\lambda$-strongly convex with $\lambda > 0$, it follows that $\rho\mapsto\mathcal{V}(\rho)$ is $\lambda$-displacement convex \citep[Proposition 9.3.2]{ambrosio2005}. In turn, this implies that $\rho\mapsto \mathcal{H}(\rho | \pi)$ is $\lambda$-displacement convex. Recall that $\mathcal{H}(\rho | \pi) \geq 0$ for any $\rho$, and that $\rho \mapsto \mathcal{H}(\rho | \pi)$ is uniquely minimized at $\pi$ due to the strict convexity of the function $x \mapsto x \log x$ for $x > 0$ appearing in $\mathcal{H}(\rho)$, and Jensen's inequality. The result in \eqref{eq:contraction_F} can then be formulated as
\begin{equation}\label{eq:contraction}
\was_2(\rho(t),\pi) \leq \was_2(\rho_0,\pi)e^{-\lambda t} \quad \text{and}\quad \mathcal{H}(\rho(t) | \pi) \leq \mathcal{H}(\rho_0 | \pi)e^{-2\lambda t}.
\end{equation}
This is a more general statement of the exponential convergence to $\pi$ of the solution to the Fokker--Planck equation mentioned in the introduction, and is as such one of the main motivations for studying Langevin Monte Carlo algorithms.

\subsection{Time discretized flows}
An important theoretical tool in establishing the existence of  gradient flows is the minimizing movement scheme, often also called the JKO scheme. For a time-step $h>0$, $k \geq 0$, and $\rho_h^0 = \rho_0$, consider the iterated minimization problems
\begin{equation}\label{eq:jko}
\rho_h^{k+1} = \argmin_{\rho \in \mathcal{P}_2(\mathbb{R}^d)} \mathcal{F}(\rho) + \frac{1}{2h}\was_2^2(\rho,\rho_h^k).
\end{equation}
Such minimizers exist and are unique under weak assumptions, such as lower semi-continuity and (strong) displacement convexity of $\mathcal{F}$ \citep[see e.g.][Proposition 4.2]{ambrosio2009existence}. Both of these conditions hold for the relative entropy functional $\rho\mapsto \mathcal H( \rho | \pi)$ when $V$ is convex: the first property holds in more generality and is well-known, whereas the second was proved in \citet{mccann1997convexity}. 

In the Euclidean setting, the sequence $(x_h^k)_{k\geq 0}$ is an implicit Euler discretization with step-size $h$ of the gradient flow of $f:\mathbb{R}^d \to \mathbb{R}$ given in \eqref{eq:grad_flow_euclidean} with initial condition $x_h^0 = x_0$ if 
\begin{equation}\label{eq:implicit_euler}
x_h^{k+1} = \argmin_{y \in \mathbb{R}^d} f(y) + \frac{1}{2h}\|x_h^{k}-y\|^2.
\end{equation}
The map defined by the right hand side of \eqref{eq:implicit_euler} is often written $\prox_f^h(x_h^k)$ in the optimization literature, and is referred to as the proximal operator \citep[see e.g.][]{parikh2014proximal}.

By analogy, the JKO scheme \eqref{eq:jko} can be seen as an implicit Euler discretization of the flow in \eqref{eq:grad_flow_wasserstein}. It was this time discretization scheme applied to the functional $\rho \mapsto \mathcal{H}(\rho | \pi)$ that \citet{jordan1998variational} employed, showing that the interpolation 
\begin{equation}\label{eq:interpolation}
\rho^h(t) = \rho_h^{k+1} \quad \text{for $t \in (kh, (k+1)h]$}
\end{equation}
converges (in some formal sense) to the solution of the Fokker--Planck equation as $h \to 0$, in the case where $V$ is smooth and satisfies certain growth conditions. 

Building on results by \citet{cepa1998problame}, \citet{ambrosio2009existence} used a minimizing movement scheme to show existence and uniqueness of the gradient flow of the relative entropy functional given any convex $V$. In particular, they show that there exists a semigroup $(P_t)_{t\geq 0}$ and a unique Markov family $\{\mathbb{P}_x : x \in \mathbb{R}^d\}$ of probability measures on $(\mathbb{R}^d)^{[0,+\infty)}$ such that $\mathbb{E}_x f(X_t) = P_t f(x)$ for all bounded Borel functions $f$ and all $x \in\mathbb{R}^d$. Moreover, it is shown that $\{\mathbb{P}_x : x \in \mathbb{R}^d\}$ is reversible with respect to $\pi$, and that $\pi$ is uniquely invariant for $(P_t)_{t\geq 0}$. Restricting $(P_t)_{t\geq 0}$ to indicator functions of Borel sets $B \in \mathcal{B}(\mathbb{R}^d)$, we define $(R_t)_{t\geq 0}$ by $R_t \rho_0(B) = \int P_t 1_B d\rho_0$. The process $\rho(t) = R_t \rho_0$ then uniquely satisfies \eqref{eq:grad_flow} and the associated properties outlined in the previous section.

After originally being introduced as a theoretical tool, there has recently been interest in developing numerical implementations of the JKO scheme for solving PDEs. Several Eulerian grid-based approaches exist, see e.g. \citet{burger2012regularized,carrillo2015finite,peyre2015entropic}. By virtue of being grid-based, these have limited application in the high-dimensional sampling setting. 

It will later be seen that Langevin-based Monte Carlo can be considered a Lagrangian scheme using a particle approximation to the gradient flow. Other Lagrangian approaches have been considered by e.g. \citet{carrillo2015numerical,benamou2016discretization,carrillo2017blob}. These methods are typically adapted to accurately solving PDEs in two or three dimensions, and do not scale well with $d$. For instance, \citet{carrillo2017blob} used the modified relative entropy functional
\begin{equation}
\mathcal{F}_\gamma(\rho) =  \int \log (\phi_\gamma * \rho) d\rho + \int V d\rho + \log Z,
\end{equation}
where $\varphi_\gamma = \gamma^{-d} \varphi(x/\gamma)$ denotes a mollifier, typically a Gaussian kernel with standard deviation $\gamma>0$.
This modification makes the functional well-behaved when evaluated at an empirical measure, with the first term providing a kernel-based estimate of the entropy of the underlying distribution. For small time steps $h$, their algorithm reduces to solving a system of ODEs to evolve the particles in the empirical measure. The application of this approach to the high-dimensional setting is limited by the kernel-based estimate of entropy.

\section{Operator splitting}\label{sec:operator_splitting}
In the previous section, we alluded to the idea that Langevin Monte Carlo numerically approximates the time discretizations used to theoretically study Wasserstein gradient flows. Before making this connection clear, we first need to introduce the concept of operator splitting. 

Consider the generic Cauchy problem
\begin{equation}\label{eq:cauchy_problem}
\frac{df}{dt} = \mathcal{A}(f), \quad f(0) = f_0,
\end{equation}
with solution given by $f(t) = S_t f_0$ in semigroup notation. In many situations, the operator $\mathcal{A}$ can be split into the sum of two simpler operators: $\mathcal{A} = \mathcal{A}_1 + \mathcal{A}_2$. Let $f_j(t) = S^j_t f_0$ for $j = 1, 2$ denote the solutions to the problems
\begin{equation}\label{eq:cauchy_problem_split}
\frac{df_j}{dt} = \mathcal{A}_j(f_j), \quad f_j(0) = f_0.
\end{equation}
One can hope to estimate the solution $f$ of \eqref{eq:cauchy_problem} via $f(t) \approx (S^2_{t/n} S^1_{t/n})^n f_0$ for some large positive integer $n$, which can be justified if a Lie--Trotter--Kato product formula of the form 
\begin{equation}\label{eq:trotter}
f(t) = \lim_{n \to +\infty} (S^2_{t/n} S^1_{t/n})^n f_0
\end{equation}
holds. The book of \citet{holden2010splitting} contains a thorough overview of such results.

Returning to the Fokker--Planck equation \eqref{eq:fp}, there is a natural split between the transport part of the equation:
\begin{equation}\label{eq:fp_transport}
\frac{d\rho}{d t} = \text{div}(\rho \nabla V), \quad \rho(0) = \rho_0,
\end{equation}
and the diffusion part:
\begin{equation}\label{eq:fp_heat}
\frac{d \rho}{d t} = \Delta \rho, \quad \rho(0) = \rho_0.
\end{equation}
In his Ph.D. thesis, \citet{stojkovic2011geometric} considers such a split for the Fokker--Planck equation with smooth drift satisfying a monotonicity property, but which is not necessarily a gradient. \citet{bowles2015weak} also consider this split for the fractional Fokker--Planck equation, where the Laplacian in the diffusion equation \eqref{eq:fp_heat} is substituted for a fractional Laplacian. In both of these works, operator splitting is introduced as a theoretical tool to establish the existence of solutions to generalized Fokker--Planck equations, but they do not consider numerical aspects nor the general case of convex $V$.

The splitting interpretation carries over to the Wasserstein gradient flow formulation, where the transport equation \eqref{eq:fp_transport} can be interpreted as the gradient flow of the potential energy functional $\rho \mapsto \mathcal{V}(\rho)$, and the diffusion equation \eqref{eq:fp_heat} can be interpreted as the gradient flow of the entropy functional $\rho \mapsto \mathcal{H}(\rho)$. We now take a brief closer look at these two gradient flows.

\subsection{The transport equation}
In addition to the formulation in \eqref{eq:grad_flow}, the gradient flow of $\rho \mapsto \mathcal{V}(\rho)$ can be characterized by the semigroup $(T_t)_{t\geq 0 }$, induced by the differential inclusion
\begin{equation} \label{eq:transport_map}
\frac{d}{dt}T_t(x)  \in  -\partial V(T_t(x)),  \quad \text{$T_0(x) = x\quad$ for all $x$ s.t. $V(x)<+\infty$}.
\end{equation}
According to Theorem 11.2.3 of \citet{ambrosio2005}, there exists a unique gradient flow of $\rho \mapsto \mathcal{V}(\rho)$ and solution to \eqref{eq:transport_map}. This gradient flow satisfies $\rho(t) = (T_t)_\#\rho_0$, where $(T_t)_\#$ denotes the push-forward map associated with $T_t$.

The corresponding JKO scheme performs minimizations of the form
\begin{equation}
\rho_h^{k+1} =  \argmin_{\rho \in \mathcal{P}_2(\mathbb{R}^d)} \mathcal{V}(\rho) + \frac{1}{2h}\was_2^2(\rho,\rho_h^{k}).
\end{equation}
By the proof of Proposition 10.4.2 in \citet{ambrosio2005}, it is clear that these steps are well-defined. Moreover, the map $\mathcal{T}_h(x) = \text{prox}_V^h(x)$ is such that $\rho_h^{k+1} =  (\mathcal{T}_h)_\#\rho_h^{k}$. Since the proximal operator satisfies 
$y = \text{prox}_V^h(x) \iff (x-y)/h \in \partial V(x)$ \citep[see e.g.][]{parikh2014proximal}, this can be seen as an implicit Euler step for the evolution of $T_t$ given in \eqref{eq:transport_map}.

\subsection{The diffusion equation}
The classical diffusion equation \eqref{eq:fp_heat}, also known as the heat equation, was first described as the gradient flow of the entropy functional $\rho \mapsto \mathcal{H}(\rho)$ on the set of densities in $\mathcal{P}_2(\mathbb{R}^d)$ by \citet{jordan1998variational}. Note that $\mathcal{H}(\rho)$ is the negative Gibbs--Boltzmann entropy of $\rho$. As pointed out in the aforementioned paper, the interpretation of the diffusion equation as the gradient flow of $\mathcal{H}$ therefore provides a natural interpretation of diffusion as the tendency of a system to maximize entropy. 

Unlike the other gradient flows we have discussed, the flow of $\rho \mapsto \mathcal{H}(\rho)$ is known in closed form: it is well-known that the solution of the diffusion equation \eqref{eq:fp_heat} is given by the density $\rho(t) = \phi_t * \rho_0$, where $\phi_t$ is the Gaussian kernel defined in \eqref{eq:gaussian_kernel}.

\section{Proximal Langevin Monte Carlo}
We are now ready to describe connections between JKO discretized gradient flows, operator splitting, and Langevin-based Monte Carlo algorithms. For a time-step $h>0$ and for $k\geq 0$, consider the iterative scheme
\begin{equation}\label{eq:numerical_steps}
\rho_h^{k+1/2} = (\mathcal{T}_h)_\# \rho_h^{k}, \quad \quad \rho_h^{k+1} = \phi_h * \rho_h^{k+1/2},
\end{equation}
which can be seen as alternating between performing a JKO step for the gradient flow of $\rho \mapsto \mathcal{V}(\rho)$ and solving the exact gradient flow of $\rho \mapsto \mathcal{H}(\rho)$. Taking instead the particle perspective,  let $X_h^{0} \sim \rho_0$ and perform
\begin{equation}\label{eq:prox_lmc}
X_h^{k+1/2} = \mathcal{T}_h(X_h^{k}) = \prox_V^h(X_h^{k}), \quad \quad X_h^{k+1} = X_h^{k+1/2} + \sqrt{2h}\eta^{k+1},
\end{equation} 
where $(\eta^{k})_{k\geq 1}$ is a sequence of independent $\mathcal{N}(0,\mathcal{I}_d)$ random variables. For each $k$, the laws of $X_h^{k+1/2}$ and $X_h^{k+1}$ are equal to $\rho_h^{k+1/2}$ and  $\rho_h^{k+1}$ respectively. A generalization of this algorithm was proposed by \citet{pereyra2016proximal} and studied further in \citet{durmus2016efficient}.

Note that $\prox_V^h(x) = x - h \nabla M_V^h(x)$, where 
\begin{equation}
M_V^h(x) = \inf_{y \in \mathbb{R}^d}\left\{V(y) + \frac{1}{2h}\|x - y\|^2\right\}
\end{equation}
is the Moreau--Yosida regularization of $V$. Moreover, in the case where $V$ is twice differentiable with positive definite Hessian $D^2V(x)$ for every $x\in\mathbb{R}^d$, it is known that $\prox_V^h(x) = x - h\nabla V(x) + o(h)$ as $h\to 0$ \citep[see e.g.][Section 3.3]{parikh2014proximal}. Hence, for small $h$, the steps in \eqref{eq:prox_lmc} can be thought of as approximating the Unadjusted Langevin Algorithm.

\subsection{Convergence analysis}
We follow the approach of \citet{clement2011trotter}, which itself is an adaptation of the methods in \citet[Chapter 4]{ambrosio2005}, to establish that the scheme in \eqref{eq:numerical_steps} satisfies a Lie--Trotter--Kato formula. We will also derive an upper bound on the 2-Wasserstein distance between the interpolation $\rho^h(t) = \rho_h^{k+1}$ for $t \in (kh, (k+1)h]$ and the gradient flow $\rho(t)$ of $\rho \mapsto \mathcal{H}(\rho | \pi)$. In turn, this allows us to bound the quantity of interest, $\was_2(\rho^h(t),\pi)$. Before stating the main results, we introduce some notation.

For any $n \geq 1$ and any $0\leq k \leq n-1$, define the quantities
\begin{equation}
\delta_{h}^{k+1} = \mathcal{V}(\rho_{h}^{k+1}) -\mathcal{V}(\rho_{h}^{k+1/2}), \qquad \Delta_{h}^{k+1} = \sum_{j=1}^{k+1} \delta_{h}^{j}.
\end{equation}
Note that $\delta_{h}^{k+1}$ can also be expressed
\begin{equation}
\delta_{h}^{k+1} = \mathbb{E}V(X+\eta) - \mathbb{E}V(X),
\end{equation}
where $X\sim \rho_{h}^{k+1/2}$ and $\eta \sim \mathcal{N}(0,2h\mathcal{I}_d)$ independently. By convexity of $V$ and Jensen's inequality, it is clear that $\delta_{h}^{k+1} \geq \mathbb{E}V(\mathbb{E}(X+\eta | X)) - \mathbb{E}V(X) \geq 0$. The next results show that controlling these quantities is sufficient to establish convergence. We also remark that if one has access to independent runs of the algorithm given in \eqref{eq:prox_lmc}, one can estimate $\delta_{h}^{k+1}$ by averaging $V(X_h^{k+1}) - V(X_h^{k})$ across those runs.

\begin{theorem}\label{theorem:converge_to_gradient_flow}
Let $(\rho^{h_m}(t))_{m\geq1}$ be a sequence of discrete solutions generated from $\rho_0$, such that $h_m\Delta_{h_m}^m \to 0$ and $h_m m \to T$ for some $T >0$, as $m \to \infty$. Then, $\rho^{h_m}(t)$ converges uniformly on $[0,T]$ to $\rho(t)$, the gradient flow of $\rho\mapsto\mathcal{H}(\rho | \pi)$ started from $\rho_0$. Moreover, if $h >0$ and $n\geq 1$ are such that $hn \leq T$, then for any $t \in [0,hn]$, 
\begin{equation}\label{eq:uniform_approx_of_gradient_flow}
\was_2 (\rho^{h}(t), \rho(t)) \leq \sqrt{6h \left(\mathcal{H}(\rho_{0} | \pi) + \Delta_{h}^{n}\right)}.
\end{equation}
\end{theorem}

The corollary below follows from combining \eqref{eq:contraction} and \eqref{eq:uniform_approx_of_gradient_flow} via the triangle inequality.

\begin{corollary}\label{cor:convergence}
Suppose $V$ is $\lambda$-strongly convex. Then, under the assumptions of Theorem \ref{theorem:converge_to_gradient_flow}, we have
\begin{equation}
\was_2 (\rho^{h}(t),\pi) \leq \sqrt{6h \left(\mathcal{H}(\rho_{0} | \pi) + \Delta_{h}^{n}\right)} +  \was_2(\rho_0,\pi)e^{-\lambda t},
\end{equation}
for any $t \in [0,hn]$, where $h> 0$ and $n\geq 1$.
\end{corollary}

\subsection{Explicit rates}\label{seq:explicit_rates}
It is clear that the rate at which $h\Delta_{h}^{n} \to 0$ as $h\to 0$ is crucial in determining the quality of the approximation $\rho^{h}(t)$. Under some assumptions on $\rho_0$ and $V$, we can obtain explicit bounds on $\Delta_{h}^{n}$ in terms of $h, n$, and $d$, as will be seen below.

Suppose $V = f + g$, where $f$ is $\lambda$-strongly convex and has Lipschitz continuous gradient, and $g$ is convex and Lipschitz. That is, assume that there exist $M(d)$ and $L(d)$ such that  for all $x,y \in \mathbb{R}^d$,
\begin{align} 
\|\nabla f(x) - \nabla f(y)\| &\leq M(d)\|x-y\| \label{eq:grad_f_lip}\\
|g(x) - g(y)| &\leq L(d)\|x-y\|, \label{eq:g_lip}
\end{align}
where the notation $M(d)$ and $L(d)$ reflects potential dependence of the Lipschitz constants on dimension. Under this assumption, we can bound  $\delta_{h}^{k+1}$ as follows:
\begin{align}
\mathbb{E}V(X+\eta) - \mathbb{E}V(X) &= \mathbb{E}[f(X+\eta) - f(X)] + \mathbb{E}[g(X+\eta) - g(X)]\\
&\leq \mathbb{E}\left[\nabla f(X)^\top \eta + \frac{M(d)}{2}\|\eta\|^2\right] +L(d) \mathbb{E}\|\eta\| \label{eq:from_nesterov}\\
&\leq M(d)hd + L(d)\sqrt{2hd},
\end{align}
where \eqref{eq:from_nesterov} follows from the basic property that
\begin{equation}
f(y) \leq f(x) +\nabla f(x)^\top(y-x) + \frac{M(d)}{2}\|x-y\|^2,
\end{equation}
for all $x,y \in \mathbb{R}^d$, see for example \citet{nesterov2013introductory}. Then, $h\Delta_{h}^n \leq M(d)hd \cdot hn + L(d)\sqrt{2hd}\cdot hn$. Hence, for any $T>0$ we could take $h_m = T/m$ and satisfy the conditions of Corollary \ref{cor:convergence}. 

Next, we can use these bounds to derive explicit rates for $n$ and $h$ that yield a desired approximation error. When selecting the initial distribution, it is not unreasonable to assume that one can choose $\rho_0$ such that $\was_2(\rho_0,\pi) = \mathcal{O}(\sqrt{d})$ and $\mathcal{H}(\rho_0 | \pi) =  \mathcal{O}(d)$. See Appendix \ref{appendix:rates} for justifications and an explicit example where these assumptions hold. 

Now, if we want $\was_2 (\rho^{h}(hn),\pi) =  \mathcal{O}(\varepsilon)$ for a threshold $\varepsilon >0$, we could require that both $h\mathcal{H}(\rho_{0} | \pi) + h\Delta_{h}^{n} =  \mathcal{O}( \varepsilon^2)$ and  $\was_2(\rho_0,\pi)e^{-\lambda hn} = \mathcal{O}(\varepsilon)$. Under the assumptions above, to ensure $\was_2(\rho_0,\pi)e^{-\lambda hn} = \mathcal{O}(\varepsilon)$, it is sufficient to take $hn = \Omega(\log(d/\varepsilon^2))$. To get $h\mathcal{H}(\rho_{0} | \pi) = \mathcal{O}( \varepsilon^2)$, one can require that $h = \mathcal{O}(\varepsilon^2/d)$. Lastly, to get $h\Delta_{h}^{n} =\mathcal{O}(\varepsilon^2)$, one can in turn require that both $M(d)hd \log(\sqrt{d}/\varepsilon) = \mathcal{O}(\varepsilon^2)$ and  $L(d)\sqrt{2hd}\log(\sqrt{d}/\varepsilon) = \mathcal{O}( \varepsilon^2)$. The former can be achieved if 
\begin{equation}\label{eq:M_order_of_n}
n = \Omega\left( \frac{d M(d) \log(\sqrt{d}/\varepsilon)^2}{\varepsilon^2}\right) \quad \text{and} \quad h = \mathcal{O}\left( \frac{\varepsilon^2}{d M(d) \log(\sqrt{d}/\varepsilon)}\right),
\end{equation}
while maintaining  $hn = \Omega(\log(d/\varepsilon^2))$. Similarly, the latter can be achieved if
\begin{equation}\label{eq:L_order_of_n}
n = \Omega\left( \frac{d L(d)^2 \log(\sqrt{d}/\varepsilon)^3}{\varepsilon^4}\right) \quad\text{and} \quad h = \mathcal{O}\left( \frac{\varepsilon^4}{d L(d)^2 \log(\sqrt{d}/\varepsilon)^2}\right),
\end{equation}
still keeping $hn = \Omega(\log(d/\varepsilon^2))$.

In the case where $g=0$ (or equivalently $L(d) = 0$) and $M(d) = \mathcal{O}(1)$, we recover the assumptions on $V$ that were made in e.g. \citet{dalalyan2017further, dalalyan2017user}. Using \eqref{eq:M_order_of_n}, we see that $n  = \Omega(d \varepsilon^{-2} \log(d \varepsilon^{-2})^2)$ iterations with a step-size of $h = \log(d/\varepsilon^2)/n$ are sufficient to achieve a 2-Wasserstein error of $\mathcal{O}(\varepsilon)$. Up to log-terms, this is the same rate as those derived for ULA in the aforementioned papers. 

In the case where $g(x) \propto \|x\|_1$ so that $L(d) = \mathcal{O}(\sqrt{d})$, we get that $n = \Omega(d^2/\varepsilon^4)$ iterations are sufficient (ignoring the log-terms). This improves upon the recent results of \citet{grappin2018thesis}, who showed that if additionally $f$ is quadratic, then $n = \Omega(d^3/\varepsilon^4)$ iterations are sufficient to yield a $2$-Wasserstein error of $\mathcal{O}(\varepsilon)$. Comparing to the remark accompanying Theorem 3 of \citet{durmus2016efficient}, our results appears less sharp than the TV bounds they derive, in which $n$ depends linearly on $d$ (up to log-terms) whenever $V$ is strongly convex. As can be seen in Appendix \ref{appendix:proofs}, this likely stems from not optimally accounting for $\lambda$-displacement convexity in Lemma \ref{lemma:bound_function_with_integral}.

\section{Conclusion}
In this paper, we have developed novel connections between the fields of Wasserstein gradient flow, operator splitting, and Langevin Monte Carlo. We have demonstrated that the gradient flow perspective allows us to derive new convergence results about a proximal version of the Unadjusted Langevin Algorithm. Under certain assumptions on the potential $V$, we derive results that are on par with the contemporary literature on ULA. However, we point out that there is room for improvement in our current proofs.  In particular, they could be improved by better accounting for the condition that $V$ is $\lambda$-strongly convex, allowing us to obtain sharper bounds when that assumption is present. On the other hand, the proof of Theorem \ref{theorem:converge_to_gradient_flow} generalizes to any convex $V$. Hence, to obtain control over the proximal ULA algorithm in such a case, one would only need to formulate conditions under which one can still derive a rate of convergence of the exact gradient flow to $\pi$, though one should no longer expect this convergence to be exponentially fast. Some recent progress in this direction  based on Lojasiewicz inequalities was made by \citet{blanchet2016family}.  

We also hope that these connections can have implications on methodology. The many other splitting schemes discussed by \citet{holden2010splitting} and in the optimization literature can potentially lead to new sampling algorithms. The same holds for other numerical schemes, such as the alternative JKO algorithm developed by \citet{legendre2017second}. For the Fokker--Planck equation, they show that their new scheme is second-order convergent, improving the original JKO scheme's first-order convergence. Recently, \citet{plazotta2018bdf2} developed a variational formulation of the BDF2 scheme applicable to the estimation of gradient flows. It is also likely that the growing literature on Langevin Monte Carlo and its variations can lead to new time discretization schemes that are of both practical and theoretical interest to the gradient flow community.

{\bf Acknowledgements}: I am greatly indebted to Nicolas Chopin and Marco Cuturi for hosting my visit to ENSAE ParisTech and CREST, where the material in this paper was developed. I'd also like to thank L\'ena\"ic Chizat, Arnak Dalalyan, Jeremy Heng, Pierre E. Jacob, Boris Muzellec and Gabriel Peyr\'e for interesting conversations about optimal transport, gradient flows, and Monte Carlo sampling. This material is based upon research supported by the Chateaubriand Fellowship of the Office for Science \& Technology of the Embassy of France in the United States.

\bibliography{biblio}
\bibliographystyle{apalike}

\appendix

\section{Proofs}\label{appendix:proofs}
Closely following \citet{clement2011trotter} and \citet{ambrosio2005}, we start by proving a discrete version of the evolution variational inequality used to characterize gradient flows. Using interpolations of the discrete solutions, we use the discrete EVI to build a continuous approximation to the desired EVI. With this approximation, we derive a bound that quantifies the closeness of two discrete solutions. This bound is used to show that under appropriate assumptions on a sequence of discrete solutions, this sequence is Cauchy and therefore has a limit. Lastly, this limit is shown to be the desired gradient flow. 

\begin{lemma}[Discrete Evolution Variation Inequality]\label{lemma:Discrete Evolution Variation Inequality}
For any $n\geq 1$, $h >0$, $\nu  \ll \mu_{\text{Leb}}$ and $k = 0, \dots, n-1$ we have
\begin{align}
\begin{split}
\frac{1}{2h}&\left[ \was_2^2(\rho_{h}^{k+1},\nu) - \was_2^2(\rho_{h}^{k},\nu) \right] +\frac{\lambda}{2}\was_2^2(\rho_h^{k+1/2},\nu) \\
&\leq \mathcal{H}(\nu | \pi) - \mathcal{H}(\rho_{h}^{k+1} | \pi) - \frac{1}{2h}\was_2^2(\rho_{h}^{k+1/2},\rho_{h}^{k}) + \delta_{h}^{k+1}.
\end{split}
\end{align}
\end{lemma}

\begin{proof}
By Corollary 4.1.3 of \citet{ambrosio2005} (see also their Lemma 9.2.7), for any $\rho_{h}^k \ll \mu_{\text{Leb}}$, we have 
\begin{align}\label{eq:discrete_var_ineq_transport}
\begin{split}
\frac{1}{2h}&\left[ \was_2^2(\rho_{h}^{k+1/2},\nu) - \was_2^2(\rho_{h}^{k},\nu) \right] +\frac{\lambda}{2}\was_2^2(\rho_h^{k+1/2},\nu)\\
&\leq \mathcal{V}(\nu) - \mathcal{V}(\rho_{h}^{k+1/2}) - \frac{1}{2h} \was_2^2(\rho_{h}^{k+1/2},\rho_{h}^{k}).
\end{split}
\end{align}

Recall that $t \mapsto \phi_t * \rho_{h}^{k+1/2}$ is the gradient flow of the 0-displacement convex entropy functional $\rho\mapsto\mathcal{H}(\rho)$. Therefore,
\begin{equation}
\frac{d}{dt} \frac{1}{2} \was_2^2(\phi_t * \rho_{h}^{k+1/2},\nu) + \mathcal{H}(\phi_t * \rho_{h}^{k+1/2}) \leq \mathcal{H}(\nu),
\end{equation}
in the sense of distributions. By Remark 1.2 of \citet{clement2011trotter}, an equivalent condition is: for all $0 < a < b< \infty$,
\begin{align}
\begin{split}
\frac{1}{2}&\left[ \was_2^2(\phi_b * \rho_{h}^{k+1/2},\nu) - \was_2^2(\phi_a* \rho_{h}^{k+1/2},\nu) \right] \\
&\leq (b-a)\mathcal{H}(\nu) - \int_a^b \mathcal{H}(\phi_t * \rho_{h}^{k+1/2}) dt.
\end{split}
\end{align}
Noting that $t \mapsto \mathcal{H}(\phi_t * \rho_{h}^{k+1/2})$ is non-increasing by Theorem 11.2.1 of \citet{ambrosio2005} (see equation 11.2.4), we have that for all $0 < a < b< \infty$,
\begin{align}
\begin{split}
\frac{1}{2}&\left[ \was_2^2(\phi_b * \rho_{h}^{k+1/2},\nu) - \was_2^2(\phi_a* \rho_{h}^{k+1/2},\nu) \right] \\
&\leq (b-a)\mathcal{H}(\nu) -  (b-a)\mathcal{H}(\phi_b * \rho_{h}^{k+1/2}).
\end{split}
\end{align}
Letting $a \to 0$, $b = h$, we have
\begin{equation}\label{eq:discrete_var_ineq_diffusion}
\frac{1}{2h}\left[ \was_2^2(\rho_{h}^{k+1},\nu) - \was_2^2(\rho_{h}^{k+1/2},\nu) \right] \leq \mathcal{H}(\nu) - \mathcal{H}(\rho_{h}^{k+1}).
\end{equation}
Adding inequalities \eqref{eq:discrete_var_ineq_transport} and \eqref{eq:discrete_var_ineq_diffusion}, as well as adding and subtracting $\mathcal{V}(\rho_{h}^{k+1})$ to the right hand side to make $\delta_{h}^{k+1}$ appear, yields the result.
\end{proof}

It can be deduced from Lemma \ref{lemma:Discrete Evolution Variation Inequality} that
\begin{equation}\label{eq:bound_wasserstein_consequtive}
\frac{1}{2h}\was_2^2(\rho_{h}^{k+1},\rho_{h}^{k}) \leq  \mathcal{H}(\rho_{h}^{k}| \pi) - \mathcal{H}(\rho_{h}^{k+1} | \pi) - \frac{1+\lambda h}{2h}\was_2^2(\rho_{h}^{k+1/2},\rho_{h}^{k}) + \delta_{h}^{k+1},
\end{equation}
by taking $\nu = \rho_{h}^{k}$, so that 
\begin{align}
\sum_{k=0}^{n-1}\was_2^2(\rho_{h}^{k+1},\rho_{h}^{k}) &\leq  2h\left[\mathcal{H}(\rho_{h}^{0} | \pi) - \mathcal{H}(\rho_{h}^{n} | \pi) + \Delta_{h}^{n}\right],\\
& \leq 2h\left[\mathcal{H}(\rho_{h}^{0} | \pi) + \Delta_{h}^{n}\right].
\end{align}
Similarly,
\begin{equation}
\was_2^2(\rho_{h}^{k+1/2},\rho_{h}^{k}) \leq \frac{2h}{1+\lambda h}\left[\mathcal{H}(\rho_{h}^{k}| \pi) - \mathcal{H}(\rho_{h}^{k+1} | \pi) + \delta_{h}^{k+1}\right],
\end{equation}
so that
\begin{equation}\label{eq:bound_wasserstein_consequtive_half}
\sum_{k=0}^{n-1}\was_2^2(\rho_{h}^{k+1/2},\rho_{h}^{k}) \leq \frac{2h}{1+\lambda h}\left[\mathcal{H}(\rho_{h}^{0} | \pi) + \Delta_{h}^{n}\right].
\end{equation}

\vspace{12pt}
Before proceeding, we introduce some more notation. Introduce the delayed interpolation $\rho_{h}(t) = \rho_{h}^{k}$ if $t \in [hk, (k+1)h)$,  and note that $\rho^{h}(t)$ and $\rho_{h}(t)$ are left and right continuous respectively. Introduce also an interpolation of the half-steps, denoted by $\rho^h_{1/2}(t) = \rho_{h}^{k+1/2}$ if $t \in [hk, (k+1)h)$.

Define the piecewise affine function 
\begin{equation} 
\ell_{h}(t) = \frac{t - hk}{h} \qquad \text{if $t \in [hk, (k+1)h)$,}
\end{equation}
and in turn let 
\begin{align}
\was^2_{h}(t, \nu) &= (1- \ell_{h}(t))\was_2^2(\rho_{h}(t), \nu) + \ell_{h}(t)\was_2^2(\rho^{h}(t), \nu), \\
\mathcal{H}_{h}(t) &= (1- \ell_{h}(t))\mathcal{H}(\rho_{h}(t) | \pi) + \ell_{h}(t)\mathcal{H}(\rho^{h}(t) | \pi). 
\end{align}
Let also
\begin{equation}
R_{h}(t) = 2(1- \ell_{h}(t))\left(\mathcal{H}(\rho_{h}^{k} | \pi) - \mathcal{H}(\rho_{h}^{k+1} | \pi) + \delta_{h}^{k+1}\right) + 2\ell_{h}(t) \delta_{h}^{k+1}
\end{equation}
for $t \in [hk, (k+1)h)$. By \eqref{eq:bound_wasserstein_consequtive} and $\delta_{h}^{k+1} \geq 0$, it is clear that $R_{h}(t)\geq 0$. The following result is an analog of Theorem 4.1.4 of \citet{ambrosio2005}.

\begin{lemma}[Gradient flow approximation]\label{lemma:Gradient flow approximation}
For any $n\geq 1$, $h >0$, $\nu \ll \mu_{\text{Leb}}$ and $t\in [0,hn]\setminus \{kh : k = 0,\dots,n\}$, we have
\begin{equation}
\frac{d}{dt}\frac{1}{2}\was_{h}^2(t,\nu) +\frac{\lambda}{2}\was_2^2(\rho^h_{1/2}(t),\nu)  + \mathcal{H}_{h}(t) - \mathcal{H}(\nu | \pi) \leq \frac{1}{2}R_{h}(t),
\end{equation}
where $d/dt$ denotes the pointwise derivative.
\end{lemma}

\begin{proof}
If $t \in  (hk, (k+1)h)$, then
\begin{equation}
\frac{d}{dt}\frac{1}{2}\was^2_{h}(t, \nu) = \frac{1}{2h}\left[ \was_2^2(\rho_{h}^{k+1},\nu) - \was_2^2(\rho_{h}^{k},\nu) \right].
\end{equation}
By Lemma \ref{lemma:Discrete Evolution Variation Inequality}, this means
\begin{align}
\frac{d}{dt}&\frac{1}{2}\was^2_{h}(t, \nu) +\frac{\lambda}{2}\was_2^2(\rho^h_{1/2}(t),\nu)  + \mathcal{H}_{h}(t) - \mathcal{H}(\nu | \pi) \\
& = \frac{1}{2h}\left[ \was_2^2(\rho_{h}^{k+1},\nu) - \was_2^2(\rho_{h}^{k},\nu) \right] +\frac{\lambda}{2}\was_2^2(\rho^h_{1/2}(t),\nu) +  \mathcal{H}_{h}(t) - \mathcal{H}(\nu | \pi) \\
&\leq \mathcal{H}_{h}(t) - \mathcal{H}(\rho_{h}^{k+1} | \pi) + \delta_{h}^{k+1} \\
& = (1- \ell_{h}(t))\mathcal{H}(\rho_{h}^k | \pi) + \ell_{h}(t)\mathcal{H}(\rho_{h}^{k+1} | \pi) - \mathcal{H}(\rho_{h}^{k+1} | \pi) + \delta_{h}^{k+1} \\
& = (1- \ell_{h}(t))\left(\mathcal{H}(\rho_{h}^k | \pi) -\mathcal{H}(\rho_{h}^{k+1} | \pi) \right)+ \delta_{h}^{k+1} \\
& =\frac{1}{2} R_{h}(t).
\end{align}
\end{proof}

\begin{lemma}\label{lemma:bound_integral}
For any $n\geq 1$, $h >0$ and $k = 0, \dots, n-1$, we have the estimate
\begin{equation}
0\leq \int_0^{(k+1)h}R_{h}(t) dt \leq h\left(\mathcal{H}(\rho_{h}^{0} | \pi) + 2\Delta_{h}^{n}\right).
\end{equation}
\end{lemma}

\begin{proof}
The lower bound follows from $R_{h}(t) \geq 0$ for all $t\in[0,hn]$. Observe that
\begin{equation}
\int_{kh}^{(k+1)h} \ell_{h}(t) dt = \int_{kh}^{(k+1)h} (1 - \ell_{h}(t)) dt = \frac{1}{2}h,
\end{equation}
which in turn implies that
\begin{align}
&\int_0^{(k+1)h}R_{h}(t) dt  = \sum_{j=0}^{k-1} \int_{jh}^{(j+1)h}R_{h}(t) dt \\
& = \sum_{j=0}^{k-1}  h \left(\mathcal{H}(\rho_{h}^{j}| \pi) - \mathcal{H}(\rho_{h}^{j+1} | \pi) + \delta_{h}^{j+1}\right) + \sum_{j=0}^{k-1}  h\delta_{h}^{j+1} \\
&\leq h\left(\mathcal{H}(\rho_{h}^{0} | \pi) - \mathcal{H}(\rho_{h}^{k+1} | \pi) + \Delta_{h}^{k+1}\right) + h \Delta_{h}^{k+1}\\
&\leq h\left(\mathcal{H}(\rho_{h}^{0} | \pi) + 2\Delta_{h}^{n}\right).
\end{align}
\end{proof}

Let $(\gamma_{r}^{j})_{j=0}^m$ denote a trajectory corresponding to another time-step ${r}$, and define the quantities $\gamma_{r}(s), \gamma^{r}(s), \ell_{r}(s), \mathcal{H}_{r}(s)$ and $R_{r}(s)$ analogously to those defined in terms of $h$. Define
\begin{equation}
\was^2_{h, r}(t,s) = (1- \ell_{r}(s))\was^2_{h}(t, \gamma_{r}(s)) + \ell_{r}(s)\was^2_{h}(t, \gamma^{r}(s)),
\end{equation}
and observe that this function is continuous in $t$ and $s$. 

\begin{lemma}\label{lemma:bound_function_with_integral}
For any $n, m \geq 1$, $h, r >0$ and $t \in [0,\min\{hn,rm\}]$,
\begin{equation}
\was^2_{h, r}(t,t) \leq \was_2^2(\rho_{h}^0, \gamma_{r}^0) +  \int_0^t R_{h}(t) +  R_{r}(t) dt.
\end{equation}
\end{lemma}

\begin{proof}
Let $s \in [0,rm]$ and $t\in[0,hn]\setminus \{kh : k = 0,\dots,n\}$. By Lemma \ref{lemma:Gradient flow approximation},
\begin{equation}
\frac{\partial}{\partial t}\frac{1}{2} \was^2_{h, r}(t,s) + \mathcal{H}_{h}(t) - \mathcal{H}_{r}(s) \leq \frac{1}{2}R_{h}(t).
\end{equation}
Similarly, for $s \in [0,rm]\setminus\{jr : j = 0,\dots,m\}$ and $t\in[0,hn]$,
\begin{equation}
\frac{\partial}{\partial s} \frac{1}{2} \was^2_{r, h}(s,t) + \mathcal{H}_{r}(s) - \mathcal{H}_{h}(t) \leq \frac{1}{2} R_{r}(s).
\end{equation}
Note the symmetry
\begin{equation}
\was^2_{h, r}(t,s) = \was^2_{r, h}(s,t),
\end{equation}
so that for $s \in [0,rm]\setminus \{jr : j = 0,\dots,m\}$ and $t\in[0,hn]\setminus\{kh : k = 0,\dots,n\}$,
\begin{equation}
\frac{\partial}{\partial t} \was^2_{h, r}(t,s) + \frac{\partial}{\partial s} \was^2_{h, r}(t,s) \leq R_{h}(t) +  R_{r}(s),
\end{equation}
by adding the inequalities above. Setting $s = t$ and letting $t \in [0,\min\{hn,rm\}]\setminus (\{kh : k = 0,\dots,n\}\cup\{jr : j = 0,\dots,m\})$,
\begin{equation}
\frac{d}{dt} \was^2_{h, r}(t,t) \leq R_{h}(t) +  R_{r}(t).
\end{equation}
Since $t \mapsto \was^2_{h, r}(t,t)$ is continuous and piecewise differentiable, the Fundamental Theorem of Calculus implies that
\begin{align}
\was^2_{h, r}(t,t) &\leq \was^2_{h, r}(0,0) + \int_0^t R_{h}(t) +  R_{r}(t) dt \\
& =  \was_2^2(\rho_{h}^0, \gamma_{r}^0) +  \int_0^t R_{h}(t) +  R_{r}(t) dt.
\end{align}
\end{proof}

\begin{lemma}\label{lemma:bound_wasserstein_two_discretizations}
For any $n, m \geq 1$, $h, r >0$ and $t \in [0,\min\{hn,rm\}]$,
\begin{align}
\begin{split}
\was^2_2 &(\rho^{h}(t), \gamma^{r}(t)) \\
& \leq 6\left[\was_2^2(\rho_{h}^0, \gamma_{r}^0) + {h}\left(\mathcal{H}(\rho_{h}^{0} | \pi) + \Delta_{h}^{n}\right) + {r}\left(\mathcal{H}(\gamma_{r}^{0} | \pi) + \Delta_{r}^{m}\right)  \right].
\end{split}
\end{align}
\end{lemma}

\begin{proof}
Suppose $j$ and $k$ are such that $t \in [kh, (k+1)h) \cap  [jr, (j+1)r)$. Then,
\begin{align*}
&\was^2_2 (\rho^{h}(t), \gamma^{r}(t))  = \was_2^2(\rho_{h}^{k+1},\gamma_{r}^{j+1}) \\
&=  (1- \ell_{h}(t))(1- \ell_{r}(t))\was_2^2(\rho_{h}^{k+1},\gamma_{r}^{j+1})  \\
&\quad +(1- \ell_{h}(t))\ell_{r}(t)\was_2^2(\rho_{h}^{k+1},\gamma_{r}^{j+1})   \\
&\quad + \ell_{h}(t) (1- \ell_{r}(t))\was_2^2(\rho_{h}^{k+1},\gamma_{r}^{j+1})  \\
&\quad + \ell_{h}(t)\ell_{r}(t) \was_2^2(\rho_{h}^{k+1},\gamma_{r}^{j+1}) \\
&\leq  3(1- \ell_{h}(t))(1- \ell_{r}(t))\left[ \was_2^2(\rho_{h}^{k+1},\rho_{h}^{k}) + \was_2^2(\rho_{h}^{k},\gamma_{r}^{j}) + \was_2^2(\gamma_{r}^{j+1},\gamma_{r}^{j})\right]  \\
&\quad +2(1- \ell_{h}(t))\ell_{r}(t)\left[ \was_2^2(\rho_{h}^{k+1},\rho_{h}^{k}) + \was_2^2(\rho_{h}^{k},\gamma_{r}^{j+1})\right]  \\
&\quad + 2\ell_{h}(t) (1- \ell_{r}(t))\left[ \was_2^2(\gamma_{r}^{j+1},\gamma_{r}^{j}) + \was_2^2(\rho_{h}^{k+1},\gamma_{r}^{j}) \right]  \\
&\quad + \ell_{h}(t)\ell_{r}(t) \was_2^2(\rho_{h}^{k+1},\gamma_{r}^{j+1}) \\
&\leq  3(1- \ell_{h}(t))(1- \ell_{r}(t))\left[ \was_2^2(\rho_{h}^{k+1},\rho_{h}^{k}) + \was_2^2(\rho_{h}^{k},\gamma_{r}^{j}) + \was_2^2(\gamma_{r}^{j+1},\gamma_{r}^{j})\right]  \\
&\quad +3(1- \ell_{h}(t))\ell_{r}(t)\left[ \was_2^2(\rho_{h}^{k+1},\rho_{h}^{k}) + \was_2^2(\rho_{h}^{k},\gamma_{r}^{j+1})\right]  \\
&\quad + 3\ell_{h}(t) (1- \ell_{r}(t))\left[ \was_2^2(\gamma_{r}^{j+1},\gamma_{r}^{j}) + \was_2^2(\rho_{h}^{k+1},\gamma_{r}^{j}) \right]  \\
&\quad + 3\ell_{h}(t)\ell_{r}(t) \was_2^2(\rho_{h}^{k+1},\gamma_{r}^{j+1}) \\
& =  3\was^2_{h, r}(t,t) + 3(1- \ell_{h}(t))\was_2^2(\rho_{h}^{k+1},\rho_{h}^{k}) + 3(1- \ell_{r}(t)) \was_2^2(\gamma_{r}^{j+1},\gamma_{r}^{j}).
\end{align*}
Now, by Lemmas \ref{lemma:bound_function_with_integral} and \ref{lemma:bound_integral}, 
\begin{align}
\was^2_{h, r}(t,t) &\leq \was_2^2(\rho_{h}^0, \gamma_{r}^0) +  \int_0^t R_{h}(t) +  R_{r}(t) dt\\
&\leq \was_2^2(\rho_{h}^0, \gamma_{r}^0) + {h}\left(\mathcal{H}(\rho_{h}^{0} | \pi) + 2\Delta_{h}^{n}\right) + {r}\left(\mathcal{H}(\gamma_{r}^{0} | \pi) + 2\Delta_{r}^{m}\right).
\end{align}
Lastly, we know by Lemma \ref{lemma:Discrete Evolution Variation Inequality} that
\begin{align}
\was_2^2(\rho_{h}^{k+1},\rho_{h}^{k}) &\leq  2{h}\left(\mathcal{H}(\rho_{h}^{0}| \pi)  + \Delta_{h}^{n}\right),\\
\was_2^2(\gamma_{r}^{j+1},\gamma_{r}^{j}) &\leq  2{r}\left(\mathcal{H}(\gamma_{r}^{0}| \pi)  + \Delta_{r}^{m}\right).
\end{align}
In conclusion, and without optimizing the constant, we get
\begin{align}
\begin{split}
\was^2_2 &(\rho^{h}(t), \gamma^{r}(t)) \\
& \leq 6\left[\was_2^2(\rho_{h}^0, \gamma_{r}^0) + {h}\left(\mathcal{H}(\rho_{h}^{0} | \pi) + \Delta_{h}^{n}\right) + {r}\left(\mathcal{H}(\gamma_{r}^{0} | \pi) + \Delta_{r}^{m}\right)  \right].
\end{split}
\end{align}
\end{proof}

Before giving its proof, we restate the main theorem of the paper:
\begin{customthm}{1}
Let $(\rho^{h_m}(t))_{m\geq1}$ be a sequence of discrete solutions generated from $\rho_0$, such that $h_m\Delta_{h_m}^m \to 0$ and $h_m m \to T$ for some $T >0$, as $m \to \infty$. Then, $\rho^{h_m}(t)$ converges uniformly on $[0,T]$ to $\rho(t)$, the gradient flow of $\rho\mapsto\mathcal{H}(\rho | \pi)$ started from $\rho_0$, as $m\to \infty$. Moreover, if $h >0$ and $n\geq 1$ are such that $hn \leq T$, then for any $t \in [0,hn]$, 
\begin{equation}
\was_2 (\rho^{h}(t), \rho(t)) \leq \sqrt{6h \left(\mathcal{H}(\rho_{0} | \pi) + \Delta_{h}^{n}\right)}.
\end{equation}
\end{customthm}

\begin{proof}
Let the discrete solutions $\rho^{h_n}(t)$ and $\rho^{h_m}(t)$ be members of the sequence. From Lemma \ref{lemma:bound_wasserstein_two_discretizations}, we know that $\was^2_2(\rho^{h_m}(t), \rho^{h_n}(t)) \to 0$ as $m, n \to \infty$, for any $t\in[0,T]$. This implies that $(\rho^{h_m}(t))_{m\geq1}$ is a Cauchy sequence. Since $(\mathcal{P}_2(\mathbb{R}^d), \was_2)$ is complete, this means that the sequence converges to a function $\rho(t)$. Since the bound in Lemma \ref{lemma:bound_wasserstein_two_discretizations} does not depend on $t$, this convergence is uniform on $[0,T]$.

Since the convergence is uniform and $\rho^{h_n}(t)$ is left continuous, then so is the limit $\rho(t)$. Moreover, since if $t\in[kh,(k+1)h)$ for some $k = 0,\dots,n-1$, 
\begin{equation}
\was_2^2(\rho^{h_n}(t),\rho_{h_n}(t)) \leq \was_2^2(\rho_{h_n}^{k+1},\rho_{h_n}^{k}) \leq 2h_n\left(\mathcal{H}(\rho_{0}| \pi)  + \Delta_{h_n}^{n}\right) \to 0 \quad \text{as $n \to \infty$}.
\end{equation}
 Hence, $\rho_{h_n}(t)$ converges to $\rho(t)$ in the same manner as $\rho^{h_n}(t)$, meaning that the limit $\rho(t)$ is right continuous also. Combining these facts, it is clear that $\rho(t)$ is continuous. 

Similarly,
\begin{equation}
\was_2^2(\rho_{h_n}(t),\rho^{h_n}_{1/2}(t)) \leq \was_2^2(\rho_{h_n}^{k+1/2},\rho_{h_n}^{k}) \leq 2h_n\left(\mathcal{H}(\rho_{0}| \pi)  + \Delta_{h_n}^{n}\right) \to 0 \quad \text{as $n \to \infty$},
\end{equation}
by the bound in \eqref{eq:bound_wasserstein_consequtive_half}. This implies that $\rho^{h_n}_{1/2}(t)$ converges to $\rho(t)$ in the same manner as $\rho_{h_n}(t)$ and $\rho^{h_n}(t)$.

It remains to show that $\rho(t)$ is the gradient flow of $\rho\mapsto \mathcal{H}(\rho|\pi)$. Indeed, let $f \in C_c^\infty((0,\infty); \mathbb{R})$ be non-negative and $\nu \ll \mu_{\text{Leb}}$. Note that $\lim_{n \to \infty} \was_{h_n}^2(t,\nu) = \was_2^2(\rho(t),\nu)$ uniformly on $[0,T]$. Since $t \mapsto \was_{h_n}^2(t,\nu)$ is continuous, so is the limit $\was_2^2(\rho(t),\nu)$. Thus, $t \mapsto f^\prime(t) \was_2^2(\rho(t),\nu)$ is continuous, i.e. integrable, on [0,T]. The continuity of $f^\prime$ implies that there exists an $M>0$ such that $|f^\prime(t)| \leq M$. In combination with the aforementioned uniform convergence, we know that 
\begin{equation}
\lim_{n \to \infty} \int_0^T f^\prime(t) \was_{h_n}^2(t,\nu)dt = \int_0^T f^\prime(t) \was_2^2(\rho(t),\nu)dt.
\end{equation}
By the same reasoning, and the fact that $\lim_{n \to \infty} \was^2_2(\rho_{1/2}^{h_n}(t),\nu) = \was_2^2(\rho(t),\nu)$ uniformly on $[0,T]$, we have
\begin{equation}
\lim_{n \to \infty} \int_0^T f(t) \was^2_2(\rho_{1/2}^{h_n}(t),\nu)dt = \int_0^T f(t) \was_2^2(\rho(t),\nu)dt.
\end{equation}

Now, since $f$ and $\mathcal{H}(\cdot | \pi)$ are non-negative, so is the function $t \mapsto f(t)\mathcal{H}_{h}(t)$. Thus, by Fatou's lemma,
\begin{equation}
\liminf_{n \to \infty} \int_0^T f(t)\mathcal{H}_{h_n}(t) dt \geq \int_0^T \liminf_{n \to \infty} f(t)\mathcal{H}_{h_n}(t)dt. \label{eq:fatous_lemma}
\end{equation}
By Lemma 2.8 of \citet{clement2011trotter}, 
\begin{equation}
\int_0^T \liminf_{n \to \infty} f(t)\mathcal{H}_{h_n}(t) dt \geq \int_0^T f(t)\mathcal{H}(\rho(t) | \pi) dt \label{eq:lemma2.8}.
\end{equation}
So,
\begin{align}
\int_0^T &\left[-f^\prime(t)\frac{1}{2}\was_2^2(\rho(t),\nu) +f(t)\frac{\lambda}{2}\was_2^2(\rho(t),\nu) + f(t)\mathcal{H}(\rho(t) | \pi) \right]dt \\
&\leq\liminf_{n \to \infty} \int_0^T \left[-f^\prime(t)\frac{1}{2}\was_{h_n}^2(t,\nu)+f(t)\frac{\lambda}{2}\was_2^2(\rho_{1/2}^{h_n}(t),\nu)  + f(t)\mathcal{H}_{h_n}(t)  \right]dt  \label{subeq:first}\\
&= \liminf_{n \to \infty} \int_0^T\left[ f(t)\frac{d}{dt}\frac{1}{2}\was_{h_n}^2(t,\nu) + f(t)\frac{\lambda}{2}\was_2^2(\rho_{1/2}^{h_n}(t),\nu)  + f(t)\mathcal{H}_{h_n}(t) \right] dt \label{subeq:second} \\
&\leq \liminf_{n \to \infty} \int_0^T f(t)\left[ \frac{1}{2}R_{h_n}(t) + H(\nu|\pi)\right] dt \label{subeq:third}\\
&= \int_0^T f(t)H(\nu|\pi)dt + \liminf_{n \to \infty} \int_0^T  f(t)\frac{1}{2}R_{h_n}(t)dt \label{subeq:fourth} \\
&\leq \int_0^T f(t)H(\nu|\pi)dt +\sup_{t\in[0,T]}f(t)\liminf_{n \to \infty} \int_0^T \frac{1}{2}R_{h_n}(t)dt \label{subeq:fifth}\\
& \leq \int_0^T f(t)H(\nu|\pi)dt +\sup_{t\in[0,T]}f(t)\liminf_{n \to \infty}\left[\frac{1}{2}h_n(\mathcal{H}(\rho_0 | \pi) + 2\Delta_{h_n}^n)\right] \label{subeq:sixth}\\
& =  \int_0^T f(t)H(\nu|\pi)dt \label{subeq:seventh},
\end{align}
where \eqref{subeq:first} follows from \eqref{eq:fatous_lemma} and \eqref{eq:lemma2.8}, \eqref{subeq:second} follows by integration by parts, \eqref{subeq:third} follows by Lemma \ref{lemma:Gradient flow approximation}, \eqref{subeq:fifth} follows by $f$ being non-negative and continuous, and $R_{h_n}(t) \geq 0$, \eqref{subeq:sixth} follows by Lemma \ref{lemma:bound_integral}, and \eqref{subeq:seventh} follows by the assumption. This concludes the proof that $\rho(t)$ is indeed the gradient flow.

Now, fix $h>0$ and $n\geq 1$ such that $hn \leq T$. Then, for any $m\geq1$,
\begin{equation}
\was^2_2(\rho^{h}(t), \rho^{h_m}(t))  \leq 6\left[h\left(\mathcal{H}(\rho_{0} | \pi) + \Delta_{h}^{n}\right) + {h_m}\left(\mathcal{H}(\rho_0 | \pi) + \Delta_{h_m}^{m}\right)  \right],
\end{equation}
for any $t\in[0,\min\{hn,h_m m\}]$ by Lemma \ref{lemma:bound_wasserstein_two_discretizations}. Taking $m \to \infty$ yields the conclusion.

\end{proof}

\section{Rates for $\was_2(\rho_0,\pi)$ and $\mathcal{H}(\rho_0 | \pi)$} \label{appendix:rates}
In this section, we provide some heuristic support for the claim that one can often assume that $\mathcal{H}(\rho_0 | \pi) = \mathcal{O}(d)$ and $\was_2(\rho_0,\pi) = \mathcal{O}(\sqrt{d})$. These assumptions can also be shown to be hold for more general settings than those we consider below. 

Let $\rho_0(x) = Z_0^{-1} e^{-V_0(x)}$, and note that
\begin{align*}
\was^2_2(\rho_0,\pi)  & = \inf_{\gamma \in \Gamma(\rho_0,\pi)} \int_{\mathbb{R}^d\times \mathbb{R}^d} \|x-y\|^2 d\gamma(x,y) \\
& \leq \int_{\mathbb{R}^d} \int_{\mathbb{R}^d} \|x-y\|^2 d\pi(x)d\rho_0(y)\\
& = \int_{\mathbb{R}^d} \|x-\bar{x}\|^2 d\pi(x) + \int_{\mathbb{R}^d} \|y-\bar{y}\|^2 d\rho_0(y) + \|\bar{x} - \bar{y}\|^2,
\end{align*}
and where $\bar{x}$ and $\bar{y}$ are the means of $\pi$ and $\rho_0$ respectively. The third term on the last line safely be assumed to be $\mathcal{O}(d)$. By Theorem 1 of \citet{durmus2016high}, the first term can be bounded by $d/\lambda$ under the $\lambda$-strong convexity assumption. Under similar assumptions on $\rho_0$, or e.g. assuming that $V_0(x) = \sum_{i=1}^d V_0^i(x_i)$, one can also defend imposing a bound of $\mathcal{O}(d)$ for second term.

Secondly, one can easily support the assumption $\mathcal{H}(\rho_0 | \pi) = \mathcal{O}(d)$ if both $V_0(x) = \sum_{i=1}^d V_0^i(x_i)$ and $V(x) = \sum_{i=1}^d V^i(x_i)$. A less restrictive condition is to assume that $0\leq V(x) - V_0(x) \leq a\|x\|^2 +b$ for some $a\geq 0$ and $b\in\mathbb{R}$ not dependent on $d$. The first inequality is analogous to saying that $\rho_0$ has heavier tails than $\pi$, whereas the second inequality constrains exactly how much heavier these tails can be. Under this assumption, and using the proof of Lemma 3 of \citet{dalalyan2014theoretical}, we can write
\begin{align*}
\mathcal{H}(\rho_0 | \pi)  & =  \int_{\mathbb{R}^d}\log \left(\frac{\rho_0}{\pi}\right) d\rho_0\\
& = \int_{\mathbb{R}^d} \left[V(x) - V_0(x)\right]d\rho_0 + \log\left( \int_{\mathbb{R}^d} e^{V_0(x) - V(x)}d\rho_0\right) \\
& \leq \int_{\mathbb{R}^d} \left(a\|x\|^2 + b\right) d\rho_0,
\end{align*}
by noting that $e^{V_0(x) - V(x)}\leq 1$ by the assumption. One can then proceed as in the last paragraph.

\subsection{Gaussian initial distribution}
Let $x^\star$ denote the minimum of $V$, and let $V_0(x) = \frac{\alpha}{2}\|x-\mu\|^2 + V(x^\star)$ with $\alpha < M(d)$, so that $\rho_0$ is a Gaussian distribution. We focus on bounding $\mathcal{H}(\rho_0 | \pi)$, as bounding the Wasserstein distance can be done as in the previous section. Then, using strong convexity,  \eqref{eq:grad_f_lip} and \eqref{eq:g_lip},
\begin{align*}
V(x) &\leq V(x^\star) + L(d)\|x - x^\star\| +\nabla f(x^\star)^\top(x-x^\star) + \frac{M(d)}{2}\|x-x^\star\|^2, \\
V(x) &\geq V(x^\star) - L(d)\|x - x^\star\| +\nabla f(x^\star)^\top(x-x^\star) + \frac{\lambda}{2}\|x-x^\star\|^2,
\end{align*}
so that
\begin{align*}
&\int_{\mathbb{R}^d} \left[V(x) - V_0(x)\right]d\rho_0 \\
&= \int_{\mathbb{R}^d} \left[ \frac{M(d)}{2}\|x-x^\star\|^2 - \frac{\alpha}{2}\|x-\mu\|^2 + L(d)\|x - x^\star\| +\nabla f(x^\star)^\top(x-x^\star) \right]d\rho_0\\
&\leq  \frac{M(d)}{2}\|\mu-x^\star\|^2 +  \frac{M(d) d}{2\alpha} - \frac{\alpha d}{2\alpha} +\nabla f(x^\star)^\top(\mu-x^\star) + L(d)\left(\int_{\mathbb{R}^d} \|x - x^\star\|^2 d\rho_0\right)^{1/2}  \\
&\leq  \frac{M(d)}{2}\|\mu-x^\star\|^2 +  \frac{(M(d)-\alpha) d}{2\alpha} +\nabla f(x^\star)^\top(\mu-x^\star) + L(d)\left(\|\mu-x^\star\|^2 +\frac{d}{\alpha}\right)^{1/2}, 
\end{align*}
and
\begin{align*}
\log \int_{\mathbb{R}^d} e^{V_0(x) - V(x)} d\rho_0 &\leq \log\left(\frac{1}{Z_{1/\alpha}} \int_{\mathbb{R}^d} e^{-\frac{\lambda}{2}\|x-x^\star\|^2+L(d)\|x - x^\star\| - \nabla f(x^\star)^\top(x-x^\star)} dx \right)\\
&\leq \log\left(\frac{1}{Z_{1/\alpha}} \int_{\mathbb{R}^d} e^{-\frac{\lambda}{2}\|x-x^\star\|^2+(L(d)+\|\nabla f(x^\star)\|)\|x - x^\star\|} dx \right)\\
&=\log\left(\frac{1}{Z_{1/\alpha}} \int_{\mathbb{R}^d} e^{-\frac{\lambda}{2}\|x-x^\star\|^2+ c\|x - x^\star\|} dx \right),
\end{align*}
where $c = L(d)+\|\nabla f(x^\star)\|$ and $Z_{1/\alpha} = \int_{\mathbb{R}^d} e^{-\frac{\alpha}{2}\|x-\mu\|^2}dx$. Furthermore,
\begin{align*}
\log\left(\frac{1}{Z_{1/\alpha}} \int_{\mathbb{R}^d} e^{-\frac{\lambda}{2}\|x\|^2+ c\|x\|} dx \right) &\leq \log\left(\frac{1}{Z_{1/\alpha}} \int_{\mathbb{R}^d} e^{-\frac{\lambda}{4}\|x\|^2+ \frac{c^2}{\lambda}} dx \right) \\
& = \log\left(\frac{Z_{2/\alpha}}{Z_{1/\alpha}} e^{\frac{c^2}{\lambda}}\right)\\
& = \frac{d}{2}\log(2) + \frac{(L(d)+\|\nabla f(x^\star)\|)^2}{\lambda}\\
& \leq \frac{d}{2}\log(2) + \frac{(L(d)+M(d)\|x^\star-x^f\|)^2}{\lambda},
\end{align*}
where $x^f$ is the minimum of $f$. Hence,
\begin{align*}
\mathcal{H}(\rho_0 | \pi) \leq  \frac{M(d)}{2}&\|\mu -x^\star\|^2 +  \frac{(M(d)-\alpha) d}{2\alpha} +\|x^\star-x^f\|\|\mu-x^\star\| + L(d)\left(\|\mu-x^\star\|^2 +\frac{d}{\alpha}\right)^{1/2} \\
&+ \frac{d}{2}\log(2) + \frac{(L(d)+M(d)\|x^\star-x^f\|)^2}{\lambda}. 
\end{align*}
Take $\alpha = \lambda$ and $\mu$ such that $\|\mu -x^\star\|^2 = \mathcal{O}(d)$, and make the safe assumption that $\|x^\star-x^f\|^2 = \mathcal{O}(d)$. If $M(d) = \mathcal{O}(1)$ and $L(d)=\sqrt{d}$ like in Section \ref{seq:explicit_rates}, we get $\mathcal{H}(\rho_0 | \pi)  = \mathcal{O}(d)$.

\end{document}